\documentclass[twocolumn, superscriptaddress]{revtex4-2}
\pdfoutput=1
\usepackage[utf8]{inputenc}
\usepackage[english]{babel}
\usepackage[T1]{fontenc}
\usepackage{hyperref}
\usepackage{mathtools}
\usepackage{amsfonts}
\usepackage{amsthm}
\usepackage{braket}
\usepackage{multirow}
\usepackage{adjustbox}
\usepackage{tikz}
\usepackage{hyphenat}
\usepackage[noabbrev, capitalize]{cleveref}
\usepackage[dvipsnames]{xcolor}

\usepackage[compat=newest]{yquant}
\useyquantlanguage{groups}
\usetikzlibrary{quotes, fit, calc}

\newtheorem*{lemma}{Lemma}
\newcommand{\R}{\mathbb{R}}
\newcommand{\C}{\mathbb{C}}
\DeclareMathOperator{\diag}{diag}

\definecolor{blue}{RGB}{0, 117, 185}
\definecolor{red}{RGB}{237, 34, 11}
\newcommand\red[1]{{\color{red}{#1}}}
\newcommand\blue[1]{{\color{blue}{#1}}}

\begin{document}

\title{Faster matrix product state preparation by exploiting symmetry\hyp induced block\hyp sparsity}

\author{Felix Rupprecht}
\email{felix.rupprecht@dlr.de}
\affiliation{Institute of Quantum Technologies, German Aerospace Center, 89081 Ulm, Germany}

\author{Sabine Wölk}
\affiliation{Institute of Quantum Technologies, German Aerospace Center, 89081 Ulm, Germany}
\affiliation{Center for Integrated Quantum Science and Technology (IQST), Ulm University, 89081 Ulm, Germany}

\begin{abstract}
    Matrix product states (MPS) serve as a key tool for studying quantum systems from chemistry and condensed\hyp matter physics, making their preparation on quantum computers an important task in interfacing classical and quantum simulation.
    Many systems of interest have $U(1)$\hyp symmetries induced by particle number and spin projection conservation, allowing to restrict the MPS tensors to be of block\hyp sparse form, a property widely used in the implementation of classical algorithms such as the density matrix renormalization group.
    We reduce the cost of fault-tolerantly preparing block\hyp sparse MPS within the standard ancilla\hyp assisted linear\hyp depth approach by implementing row and column permutations that transform the block\hyp sparse matrices into block\hyp diagonal form.
    These block\hyp diagonal unitaries are then implemented via unitary synthesis, with the cost being determined by the size of the largest block.
    In this context, we modify the unitary synthesis approach of Berry et al. in order to reduce the Toffoli cost for real\hyp valued unitaries by a factor of $\sqrt{2}$.
    In numerical benchmarks, we achieve Toffoli cost improvement factors of $10 - 30$ compared to the state\hyp of\hyp the\hyp art for MPS of various molecular systems.
\end{abstract}
\maketitle

\section{Introduction}
\label{sec:introduction}

The main difficulty in direct simulations of fermionic many\hyp body systems from chemistry and condensed\hyp matter physics is the exponential growth of the corresponding Hilbert spaces with increasing particle numbers.
In order to allow computational investigations of such systems, even in the presence of strong correlation, effective wave function representations that restrict to subspaces of the full Hilbert spaces have been developed, and sophisticated methods for calculations with such compressed states implemented~\cite{Eriksen2020}.
One of the most prominent compressed wavefunction formats are matrix product states (MPS), which follow an entanglement area law and efficiently parameterize the ground states of local one\hyp dimensional gapped Hamiltonians~\cite{Cirac2021}.
Closely related to the notion of MPS is the density matrix renormalization group (DMRG) algorithm~\cite{schollwock2011}, a widely used approach for finding approximate ground states of strongly correlated many\hyp body systems~\cite{zhai2026, liu2025, legeza2026}.

Due to their conceptual and practical relevance, MPS have taken a prominent role at the interface of classical and quantum computation, with the task of preparing MPS on quantum computers gaining more attention in recent years~\cite{Fomichev2024, berry2024, kottmann2026, Schoen2005, Malz2024, Ran2020, Smith2024, Rudolph2022, gibbs2026}.
While some of the papers consider the preparation task for particular subclasses of MPS~\cite{Malz2024, Smith2024}, approximate preparation for low T-gate counts~\cite{gibbs2026} or methods well\hyp suited for NISQ devices~\cite{Rudolph2022}, the standard method for the preparation of generic MPS is the sequential, ancilla\hyp assisted technique proposed in~\cite{Schoen2005}.
Motivated by the problem of preparing high\hyp quality initial states~\cite{Lee2023} for ground\hyp state energy estimation algorithms such as quantum phase estimation (QPE), the resource requirements for this approach in the framework of fault\hyp tolerant quantum computing were analyzed and optimized in~\cite{Fomichev2024, berry2024, kottmann2026}.
Besides ground\hyp state energy estimation, MPS have also been investigated in the context of quantum machine learning~\cite{Dilip2022, Jobst2024}, variational algorithms~\cite{Rudolph2023, Haghshenas2022}, and NISQ applications~\cite{Chertkov2022, Anselme2024}.

In resource analyses of applying QPE to relevant molecular systems, a common assumption~\cite{Lee2021, tubman2018} is that the cost of preparing high\hyp quality initial states is insignificant compared with the cost of the actual phase estimation part of the algorithm.
In \cite{berry2024}, this question was investigated more thoroughly for the FeMo\hyp cofactor with MPS from classical DMRG calculations chosen as the initial states.
The Toffoli costs of preparing the MPS were indeed of smaller magnitude than the cost of the QPE part, however, recent algorithmic advances have reduced or even closed this gap~\cite{Low2025}.
In the pursuit of further bringing down the cost of end\hyp to\hyp end ground\hyp state energy estimation and allowing the use of higher-quality initial states, it is therefore important to increase the efficiency of MPS state preparation, which we aim to do in this paper.

We follow the approach of~\cite{Fomichev2024, berry2024, kottmann2026, Schoen2005}, in which a matrix product state
\begin{equation*}
    \ket{\Phi} = \sum_{d_1, \dots, d_{n}} M_{1}^{d_1} M_{2}^{d_2}\cdots M_{n}^{d_{n}} \ket{d_1,\dots,d_{n}}
\end{equation*}
is implemented via a quantum circuit of the form shown in \cref{fig:linear_mps_prep}.
The MPS wave function for a system of $n$~sites, where each site ${i\in [n]\coloneqq\;}\{1,\dots,n\}$ is equipped with a local basis $(\ket{d_i^{1}}, \dots, \ket{d_i^{D}})$ of size~$D$, stores the coefficients of the multi\hyp site product basis states~$\ket{d_1,\dots,d_n}$ as a product of matrices~$M_i^{d_i}$.
Here $M_i\coloneq (M_{i})_{b_{i}, b_{i+1}}^{d_i}$ is a rank-3 tensor of shape~$(B_{i-1}, D, B_i)$, where $B_{i-1}$ and $B_i$ are called the bond dimensions of $M_i$ and $B_0 = B_n = 1$.
The maximal bond dimension of the MPS is denoted by~$\chi$.
For ease of exposition, we assume throughout the paper that all non\hyp trivial bond dimensions~$B_i$ are equal to~$\chi$.

Before preparing the MPS on a quantum computer, it is transformed into right\hyp canonical form, i.e., the tensors satisfy $\sum_{d_i=1}^D M_i^{d_i}(M_i^{d_i})^{\dagger} = I$ for all $i \in [n]$.
This can always be achieved~\cite{schollwock2011} by means of repeated singular-value transformations, as MPS have a gauge degree of freedom, i.e., for any invertible matrix $X$, replacing $M_i^{d_i}$ and $M_{i+1}^{d_{i+1}}$ by $M_i^{d_i}X$ and $X^{-1}M_{i+1}^{d_{i+1}}$ leaves the MPS invariant.
Due to the right\hyp canonical form, the columns of the matrix
\begin{equation}
\label{eq:site_unitary}
    U_i' \coloneq \begin{bmatrix} M_i^{1} & \cdots & M_i^D \end{bmatrix}^\top\text,
\end{equation}
which for $1 < i < n$ has the shape $(d\chi)\times \chi$, are orthonormal and therefore may be expanded to a unitary
\begin{equation*}
    U_i \coloneq \begin{bmatrix} U_i' & * \end{bmatrix}
\end{equation*}
with $*$ being a placeholder submatrix.
The MPS is then prepared (cf. \cref{fig:linear_mps_prep}) on $n$ site registers, each consisting of $\lceil \log(D) \rceil$ qubits, with the help of a $\lceil \log(\chi) \rceil$\hyp qubit ancilla register where here and in the rest of the paper all logarithms are base-2.
As both the site register and the ancilla register on which the first unitary~$U_1$ acts are in the $\ket{0}$\hyp state, the implementation of $U_1$ is reduced to state preparation~\cite{berry2024, Low2024} of the state $U_1 \ket{0}\ket{0}$.
The remaining unitaries~$U_i$ are implemented via unitary synthesis, where, due to the incoming site registers being $\ket{0}$, only the values in the first $\chi$~columns, i.e., the submatrix~$U'_i$, are relevant.
It was explained in~\cite{berry2024} that the last unitary~$U_{n}$ may be integrated into the circuit of $U_{n-1}$.
However, for simplicity, we do not apply this in our work.
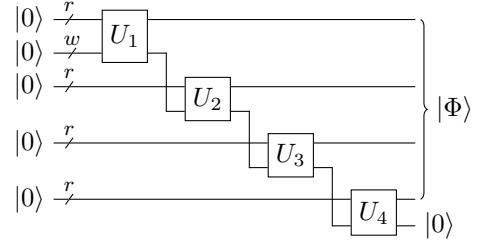
\begin{figure}[t]
    \centering
    \begin{tikzpicture}
        \begin{yquant}[register/separation=0mm]
            qubit {$\ket{0}$} site0;
            [name=anc0] qubit {$\ket{0}$} anc0;
            qubit {$\ket{0}$} site1;
            [name=anc1] nobit anc1;
            qubit {$\ket{0}$} site2;
            [name=anc2] nobit anc2;
            qubit {$\ket{0}$} site3;
            nobit anc3;

            ["north:$r$" {font=\protect\footnotesize, inner sep=1pt}]
            slash site0, site1, site2, site3;
            ["north:$w$" {font=\protect\footnotesize, inner sep=1pt}]
            slash anc0;

            \yquantset{operator/separation=2.5mm}
            box {$U_1$} (site0, anc0);
            hspace {2.5mm} anc0, anc1;
            discard anc0;
            [name = i0, operator/separation=-2.5mm] init anc1;
            \draw[shorten <=-.5\pgflinewidth, shorten >=-.5\pgflinewidth] (anc0 -| i0) -- (i0);

            [operator/separation=0pt]
            box {$U_2$} (site1, anc1);
            hspace {2.5mm} anc1, anc2;
            discard anc1;
            [name = i1, operator/separation=-2.5mm] init anc2;
            \draw[shorten <=-.5\pgflinewidth, shorten >=-.5\pgflinewidth] (anc1 -| i1) -- (i1);

            [operator/separation=0pt]
            box {$U_3$} (site2, anc2);
            hspace {2.5mm} anc2, anc3;
            discard anc2;
            [name = i2, operator/separation=-2.5mm] init anc3;
            \draw[shorten <=-.5\pgflinewidth, shorten >=-.5\pgflinewidth] (anc2 -| i2) -- (i2);

            [operator/separation=0pt]
            box {$U_4$} (site3, anc3);

            output {$\ket{0}$} anc3;
            output {$\ket{\Phi}$} (-site3);
        \end{yquant}
    \end{tikzpicture}
    \caption{
    Preparation circuit~\cite{Fomichev2024, berry2024, kottmann2026, Schoen2005} of a four\hyp site MPS~$\ket{\Phi}$ with maximal bond dimension~$\chi$ and site dimension~$D$.
    The second register is an ancilla register of size~$w\coloneq \lceil \log(\chi) \rceil$, while the other registers, consisting of $r\coloneq \lceil \log(D) \rceil$~qubits each, represent the sites. 
    Only the first $\chi$ columns of the unitaries~$U_i\coloneq \begin{bmatrix} U_i' & * \end{bmatrix}$ are specified in~\eqref{eq:site_unitary}, while the remaining entries can be chosen freely.
    By construction, the ancilla register is brought back to $\ket{0}$ at the end.
    }
    \label{fig:linear_mps_prep}
\end{figure}

In many problems of interest, such as fermionic lattice models~\cite{Kan2025} and electrons in molecules~\cite{berry2024}, the total particle number~$\hat{N}=\sum_i \hat{N}^i$ and the spin\hyp projection/magnetization~$\hat{S}_z=\sum_i \hat{S}_z^i$ both commute with the Hamiltonian and are therefore good quantum numbers corresponding to $U(1)$\hyp symmetries.
By selecting a joint eigenbasis of the local number and spin projection operators~$(\hat{N}^i, \hat{S}_z^i)$ for the Hilbert spaces of the sites~$i$, an MPS can be made invariant under these symmetries by constraining its tensors with a local quantum number conservation rule~\cite{schollwock2011, singh2011, Bachmayr2022}.
As we will recall in \cref{sec:mps_preparation}, the rule, which can be formulated for any global abelian symmetry, leads to the tensors~$M_i$ and consequently $U_i'$ having a block\hyp sparse structure.
This has been widely used for improving the efficiency of MPS algorithms~\cite{zhai2023, Hasschild2018, fishman2022} such as DMRG, raising the question of whether the block\hyp sparsity may also be exploited to improve the cost of MPS preparation on quantum computers.

The contributions of this paper are as follows.
We answer the question from above in the affirmative by using the block\hyp sparsity of $M_i$ to find more efficient ways of implementing the unitaries~$U_i$.
As is usual in fault\hyp tolerant resource estimation, we take the number of Toffoli gates together with the number of required qubits as the figures of merit.
The main idea, explained in \cref{sec:mps_preparation}, is to exploit the block\hyp sparsity of $U_i'$ by implementing column and row permutations, together with selecting suitable values for the free entries $*$, such that the matrix~$U_i = \begin{bmatrix} U'_i & *\end{bmatrix}$ is transformed into a block\hyp diagonal unitary.
The resulting matrix is then implemented via unitary synthesis, with the Toffoli cost mostly determined by the size of the largest block.
Since for most systems of interest the MPS matrices are real\hyp valued, in \cref{sec:unitary_synthesis} we modify the unitary synthesis method of Berry et al.~\cite{berry2024} based on Givens rotations to reduce the number of Toffoli gates required for real unitaries by a factor of $\sqrt{2}$.
We benchmark the overall preparation approach on various molecular systems, such as active spaces of the P450 enzyme~\cite{Goings2022}, yielding Toffoli cost improvement factors of $10 - 30$.

The algorithms of the paper are implemented within the \texttt{Qualtran}~\cite{harrigan2024} framework, with the computationally more expensive Givens decomposition outsourced into Rust code.
The full code and other assets created for this paper are available on Zenodo~\cite{Rupprecht2026Zenodo}.

\section{Unitary synthesis via Givens rotations}
\label{sec:unitary_synthesis}

In this section, we develop a variation of the unitary synthesis approach of Berry et al.~\cite{berry2024} that allows for a reduction in the number of Toffoli gates in the case of real unitaries by roughly a factor of $\sqrt{2}$.

Given a unitary matrix~$U$ of dimension~$l$, the task of (exact) unitary synthesis is to find a sequence of quantum gates~$C_1, \dotsc, C_k$ such that $U=C_k \dotsm C_1$.
A standard building block for creating such sequences are Givens rotation matrices
\begin{equation*}
    G(\theta, \phi) \coloneq
    \begin{bmatrix}
        e^{i\phi} \cos(\theta) & -\sin(\theta)\\
        e^{i\phi} \sin(\theta) & \cos(\theta)
    \end{bmatrix}
\end{equation*}
with $\phi, \theta \in \R$.
The unitary Givens rotations have the property that for $a,b \in \C$, there exist angles~$\theta$ and $\phi$ such that $\begin{bmatrix} a & b \end{bmatrix} G(\phi, \theta) = \begin{bmatrix} \tilde{a} & 0\end{bmatrix}$
for some $\tilde{a} \in \C$.
If $a$ and~$b$ are real, $\phi$ can be chosen to be zero.
By using a suitable sequence~$G_1, \dotsc, G_k$ of Givens rotations, each being embedded in an $l$\hyp dimensional matrix and acting on neighboring columns of~$U$, all off\hyp diagonal elements of $U$ can be zeroed, and $U=\tilde{D}\, G_k\dotsm G_1$ for some diagonal phase matrix~$\tilde{D}$.
This is a standard approach in numerical linear algebra for computing QR\hyp decompositions~\cite{golub}.

Since Givens rotations may be interpreted as beam splitters, the problem of unitary synthesis via Givens rotations can be seen as the task of building a multiport interferometer from beam splitters.
In this context, it was found in~\cite{Clements2016} that the sequence of Givens rotations for diagonalization may be chosen in a way that allows applying a group of rotations in parallel, resulting in a decomposition~$U= \tilde{D} \tilde{L}_{l} \dotsm \tilde{L}_1$ of $l$ layers of parallel Givens rotations and a final diagonal phase matrix~$\tilde{D}$.
The layers~$\tilde{L}_i$ alternate between the two forms
\begin{equation}
\label{eq:form}
    \begin{bsmallmatrix}
        *&*&&&&&&\\
        *&*&&&&&&\\
        &&*&*&&&&\\
        &&*&*&&&&\\
        &&&&*&*&&\\
        &&&&*&*&&\\
        &&&&&&*&*\\
        &&&&&&*&*\\
    \end{bsmallmatrix},
    \quad
    \begin{bsmallmatrix}
        *&&&&&&&\\
        &*&*&&&&&\\
        &*&*&&&&&\\
        &&&*&*&&&\\
        &&&*&*&&&\\
        &&&&&*&*&\\
        &&&&&*&*&\\
        &&&&&&&*\\
    \end{bsmallmatrix}
\end{equation}
where each $2\times 2$ block is a Givens rotation and the lonely stars are~$1$.
Due to its low depth, this decomposition is not only useful in optics but was also used in \cite{berry2024} to reduce the Toffoli cost of unitary synthesis compared to prior methods based on Householder reflections~\cite{Low2024}.

Motivated by the interferometer analogy, in~\cite{berry2024} the beam splitters corresponding to a layer of Givens rotations are decomposed into two 50:50\hyp beam splitters with a phase shift in between, followed by a diagonal phase matrix, which can be incorporated into the next layer.
Instead of following this approach, we write a Givens layer in terms of $R_z$ and $R_y$ gates
\begin{equation*}
    R_z(\phi) \coloneq
    \begin{bmatrix}
        e^{-i\phi}  & 0\\
        0           & e^{i\phi}
    \end{bmatrix},
    \,
    R_y(\theta) \coloneq
    \begin{bmatrix}
        \cos(\theta) & -\sin(\theta)\\
        \sin(\theta) & \cos(\theta)
    \end{bmatrix}
\end{equation*}
via
\begin{equation}
\label{eq:decomp}
    G(\theta, \phi) \diag(e^{i\beta_1}, e^{i\beta_2})= \diag(e^{i\beta'}, e^{i\beta'}) R_y(\theta) R_z(\phi')
\end{equation}
with $\beta'\coloneq(\phi + \beta_1 + \beta_2)/2$ and $\phi'\coloneq (\beta_2 - \beta_1 - \phi)/2$, where $\diag(e^{i\beta_1}, e^{i\beta_2})$ is a potential diagonal phase matrix from applying the same decomposition on the Givens layer to the right.
Note that the number of variables on the right and left side of~\eqref{eq:decomp} differ as $\beta_1$ can be made zero by including it in $\phi$.
Sequential application of relation~\eqref{eq:decomp} allows one to write
\begin{equation}
\label{eq:lay}
    U=D L_{l} \cdots L_{1}
\end{equation}
with the $L_j$ alternating between the two forms in~\eqref{eq:form}, where each $2\times 2$ block is given by $R_y(\theta)R_z(\phi')$.
The matrix~$D$ is again a diagonal phase matrix.
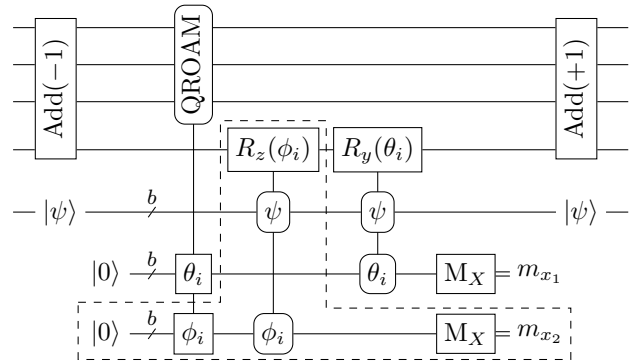
\begin{figure}[b]
    \centering
    \begin{adjustbox}{width=\columnwidth}
    \begin{tikzpicture}
        \begin{yquant}[register/separation=2mm, operator/separation=2mm]
            qubit {} system[3];
            qubit {} target;
            qubit {} phasegrad;

            hspace {0.1cm} -;
            box {\rotatebox{90}{Add($-1$)}} (system, target);
            text {$\ket{\psi}$} phasegrad;
            hspace {0.6cm} -;
            [name = init]
            [after=target] qubit {$\ket{0}$} angles[2];

            ["north:$b$" {font=\protect\footnotesize, inner sep=1pt}]
            slash angles, phasegrad;

            [name = load, operator style={only at={0}{rounded corners}}]
            box {\Ifnum\idx<1\rotatebox{90}{QROAM}\Else{\Ifnum\idx<2$\theta_i$\Else$\phi_i$\Fi}\Fi} (system), angles;

            [name = rz, operator style={only at={1,2}{rounded corners}}]
            box {\Ifcase\idx\space$R_z(\phi_i)$\Or$\psi$\Else$\phi_i$\Fi} target, phasegrad, angles[1];

            [name = ry, operator style={only at={1,2}{rounded corners}}]
            box {\Ifcase\idx\space$R_y(\theta_i)$\Or$\psi$\Else$\theta_i$\Fi} target, phasegrad, angles[0];

            box {M$_X$} angles;
            settype {cbit} angles;
            [name = discard]
            text {$m_{x_{\The\numexpr\idx+1\relax}}$} angles;
            discard angles;

            hspace {-0.4cm} -;
            box {\rotatebox{90}{Add($+1$)}} (system, target);
            text {$\ket{\psi}$} phasegrad;
            hspace {0.1cm} -;
        \end{yquant}
        \draw (load-0) -- (load-1) -- (load-2);
        \draw (rz-0) -- (rz-1) -- (rz-2);
        \draw (ry-0) -- (ry-1) -- (ry-2);
        \draw[dashed]
            ($(init-1.north west)       + (-.1,  .25)$)
            -| ($(rz-0.north west)      + (-.1,  .1)$)
            -- ($(rz-0.north east)      + ( .1,  .1)$)
            |- ($(discard-1.north east) + (  0,  .1)$)
            -- ($(discard-1.south east) + (  0, -.1)$)
            -| cycle;
    \end{tikzpicture}
    \end{adjustbox}
    \caption{
    Circuit for implementing a layer of $R_z$/$R_y$ rotations in \eqref{eq:lay} of the forms \eqref{eq:form} on four qubits up to sign flips from $X$\hyp measurements~$M_X$ of the angle registers and the QROAM junk registers.
    The measurement results of the angle registers are $m_{x_1}$ and $m_{x_2}$.
    We do not show the QROAM junk registers.
    The bottom two registers are the ancilla angle registers into which the rotation angles are loaded via QROAM.
    The rotations use a phase-gradient state~$\ket{\psi}$ as a resource state together with the loaded angles.
    We use rounded corners for gates on qubits that are not being changed by the gates and merely serve as address/resource states.
    As printed, the right layer in~\eqref{eq:form} is implemented.
    Removing the decrement/increment operators yields the left layer in~\eqref{eq:form}. 
    For real unitaries, the $R_z$ rotations are not needed, and the elements within the dotted area can be removed.
    }
    \label{fig:givens_layer}
\end{figure}

For real-valued $U$, as mentioned above, the angle $\phi$ in the Givens rotations of the layers $\tilde{L}_j$ can be chosen to be zero.
As $G(\theta, 0)=R_y(\theta)$ it is therefore possible to remove the $R_z$~gates in the decomposition~\eqref{eq:lay}, which then solely consists of layers of $R_y$~rotations with a final diagonal phase matrix~$D$ whose diagonal entries are $\pm 1$.

We implement these $R_y$ layers in the real case, and $R_z$/$R_y$ layers in the general case, within the well\hyp known phase gradient framework~\cite{berry2024, kottmann2026, Gidney2018, Kitaev2002}, as depicted in \cref{fig:givens_layer}.
For any layer, the first step in doing so is to load the rotation angles for the individual $2\times 2$ blocks.
We first look at the left matrix in~\eqref{eq:form}.
Working within a $\lceil \log(l) \rceil$\hyp qubit register, with suitably padded unitaries in case $l$ is not a power of two, we split off the last qubit~$q$.
This qubit defines the two\hyp dimensional subspaces of the blocks, while the state~$\ket{i}^r$ of the remaining qubits enumerates the blocks.
We can thus load $b$\hyp bit binary representations of the rotation angles~$\phi_i$ and $\theta_i$ for the individual blocks into ancilla registers using QROAM~\cite{Berry2019} with $\ket{i}^r$ as the address states.
In the real case, it suffices to load the angles~$\theta_i$.

Once the angles are loaded, the $R_y$ and, if necessary, $R_z$~rotations on qubit~$q$ can be implemented~\cite{kottmann2026, Sanders2020} via controlled addition/subtraction of the angle states and a $b$\hyp qubit phase-gradient state~$\ket{\psi}$.
The desired rotation axis is selected by a suitable change of the computational basis.
A rotation has Toffoli cost~$b-2$ and does not change the phase-gradient state~$\ket{\psi}$, whose register size~$b$ is logarithmic in the accuracy of the stored rotation angles and, in most practical situations, is on the order of $15$–$20$.
For the definition and a way of preparing~$\ket{\psi}$ we refer to~\cite{Sanders2020}.
As it is left unchanged by the rotations it only needs to be prepared once.
After the rotations have been applied, the angle registers are measured in the $X$~basis and discarded, introducing phase flips as described in Appendix~\labelcref{sec:measurement_based_uncomputation}.
The phase flips can be reconstructed from the measurement results and will be fixed later.

The Toffoli cost of the angle loading together with the rotation is given by $\lceil l/(2\Lambda) \rceil + 2b\Lambda$, where $\Lambda$ is a power of two that acts as the control parameter for the Toffoli/ancilla trade-off for the QROAM~\cite{Berry2019}.
Here and in the rest of the paper we neglect some constant minor terms such as the $-2$ of the rotations and rather use the upper bound stated as this makes the formulas more compact.
Moreover, accurately accounting for some of those terms is sometimes error-prone.
For a given~$\Lambda$, the number of ancilla qubits required for the QROAM is given by $2b(\Lambda -1) +  \lceil \log(l/\Lambda) \rceil$.
The optimal Toffoli cost in the trade-off is achieved for $\Lambda \approx \sqrt{l/(4b)}$, with a Toffoli count of about $\sqrt{4bl}$.
The costs in the real case, in which the bitsize of the QROAM target register is only half that of the general case, are then simply given by replacing $b$ with $b/2$ in the formulas leading to an optimal Toffoli cost of $\sqrt{2bl}$.

In order to implement the right matrix in~\eqref{eq:form}, we follow~\cite{berry2024} and decrement the full register~$\ket{j}\mapsto \ket{j-1}$ with Toffoli cost~$\lceil \log(l)\rceil -1$, before applying the same logic as for the left matrix and incrementing the register again.
The pair of decrement/increment operations may be interpreted as applying row and column permutation matrices that shift the rows and columns by one, transforming the right matrix in~\eqref{eq:form} to the left matrix with the fourth block set to the identity.
\begin{figure*}[t]
    \centering
    \begin{adjustbox}{width=\textwidth}
        \includegraphics{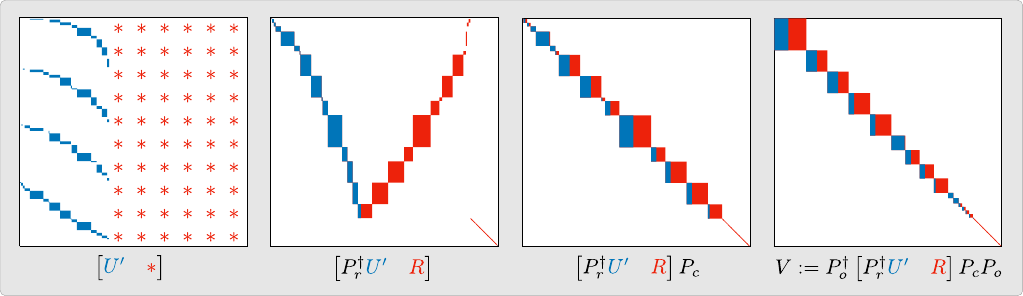}
    \end{adjustbox}
    \caption{
    Steps for transforming and extending the block\hyp sparse matrix $U'\coloneq U_i'$, represented by the blue part in the leftmost box, into a block\hyp diagonal unitary $V$ using permutation matrices.
    The content of the red rectangles~$R$, which are put into the initially unspecified part $*$ of the matrix in the second step, is chosen such that the emerging square blocks are unitary.
    }
    \label{fig:transform_to_block_diagonal}
\end{figure*}

It remains to look at the final diagonal phase matrix~$D \coloneq \diag(e^{i\gamma_1}, \dotsc, e^{i\gamma_l})$.
In the complex case, using again the phase gradient framework, $l$~phase angles with bitsize~$b$ are loaded, after which an $R_z$ rotation with these angles on an ancilla $\ket{1}$\hyp qubit is applied.
This is explained in more detail in~\cite{kottmann2026}.
Its cost is again dominated by the angle loading and is given by $\lceil l/(2\Lambda) \rceil + 2b\Lambda + b$, where, as in~\cite{berry2024}, we use the fact that, up to the trailing $b$, the ancilla/Toffoli trade-off for loading $l/2$~items of size~$2b$ is the same as for loading $l$~items of size~$b$, letting us reuse the $\Lambda$ from the layers.
Analogously to the layer rotations, the angle register is measured in the $X$~basis, causing sign flips that need to be fixed.
Since in the real case the matrix~$D$ is diagonal with entries~$\pm 1$, the entries themselves may be seen as sign flips, and no phase rotation is required.

In order to fix the sign flips, both from the layers and the phase matrix, we notice that the sign flips in layer~$L_i$ may be commuted through the subsequent layers~$L_{i+1}, \dotsc, L_l$ by dynamically adapting the angles of the $R_y$ rotations in those layers according to the relation
\begin{equation*}
   R_y(\theta)\diag(-1, 1)=\diag(-1, 1)R_y(-\theta).
\end{equation*}
The $R_z$-gates commute with the diagonal sign matrix.
At the end, the commuted sign flips can be combined with the sign flips from implementing $D$ and fixed using the standard approach from~\cite{Berry2019}.
The Toffoli cost for this is $\lceil l/\Lambda' \rceil + \Lambda'$, where $\Lambda'$ is again a Toffoli/ancilla trade-off parameter that is a power of two controlling the ancilla cost of $\Lambda' + \lceil \log(l/\Lambda')\rceil$.

Summing all up, the overall Toffoli cost of our unitary synthesis approach is upper bounded by
\begin{equation*}
    (l+1) \left( \Bigl \lceil \frac{l}{2\Lambda} \Bigr \rceil + 2b\Lambda \right) + b + \Bigl \lceil \frac{l}{\Lambda'} \Bigr \rceil + \Lambda' + l \lceil \log(l)\rceil
\end{equation*}
in the general case, and
\begin{equation*}
    l \left( \Bigl \lceil \frac{l}{2\Lambda} \Bigr \rceil + b\Lambda \right) + \Bigl \lceil \frac{l}{\Lambda'} \Bigr \rceil +  \Lambda' + l \lceil \log(l)\rceil
\end{equation*}
in the real case.
While the Toffoli cost of the factor~$l$ in the big bracket, corresponding to the individual QROM calls, can be almost quadratically compressed via $\Lambda$ by using ancilla qubits, the leading $l$, which counts the number of layers, cannot.
If the unitary~$U$ has structure, it is therefore preferable to exploit this structure to reduce the number of layers.

While the resource requirements for the general case coincide with those for the 50:50\hyp beam splitter approach in~\cite{berry2024}, our approach allows for reducing the Toffoli cost for real unitaries common in tasks such as initial state preparation by a factor of $\sqrt{2}$, as can be seen from the optimal Toffoli cost of a single layer.

\section{Preparation of block\hyp sparse MPS}
\label{sec:mps_preparation}

In this section, we show how the block\hyp sparse structure of MPS of systems with global $U(1)$\hyp symmetries can be exploited for more efficient preparation circuits.
We do this within the standard MPS preparation scheme explained in \cref{sec:introduction}.

For simplicity, we restrict ourselves to the most common case of a fermionic system consisting of $n$~sites with fixed particle number~$N$ and spin\hyp projection number~$S_z$, where the Hilbert spaces of the sites~${i\in [n]}$ are spanned by a joint eigenbasis~${(\ket{d_i^j})_{j\in [4]} \coloneq (\ket{0, 0}, \ket{1, +\frac{1}{2}}, \ket{1, -\frac{1}{2}}, \ket{2, 0})}$ of the local particle number and spin projection operators~$\hat{N}^i$ and $\hat{S}_z^i$.
Prominent examples of such systems are fermionic lattice models, such as the Fermi-- Hubbard model, or the electronic structure of molecular systems.
We note that the procedure works whenever we have a global abelian symmetry.

In the following, we briefly recall how the $U(1)$\hyp symmetries induced by $\hat{N}$ and $\hat{S}_z$ lead to the block\hyp sparsity of the MPS tensors~\cite{schollwock2011, singh2011}.
For this let $\hat{Q}\coloneq \sum_i \hat{Q}^i$ be any of the operators $\hat{N}$ and $\hat{S}_z$.
Given an eigenstate~$\ket{x}$ of $\hat{Q}_i$ or more generally $\hat{Q}^{<i} \coloneq \sum_{k=0}^{i-1} \hat{Q}^k$ we denote the corresponding quantum number by $q(x)$.
As the operators $\hat{N}$ and $\hat{S}_z$ commute and are therefore compatible we may apply the argument to both operators in parallel, i.e., when talking about an eigenvector of $\hat{Q}_i$ one may think of it as a joint eigenvector of $\hat{N}^i$ and $\hat{S_z}^i$.
Note that more generally, the operator $\hat{Q}= \sum_i \hat{Q}^i$ may be any operator that gives rise to a global abelian symmetry.

We consider a site~$i>1$ and, for $u\in [\chi]$, look at the state
\begin{equation*}
    \ket{u} \coloneq \sum_{d_1, \dotsc, d_{i-1}} (M_{1}^{d_1} M_{2}^{d_2}\dotsm M_{i-1}^{d_{i-1}})_u \ket{d_1,\dotsc,d_{i-1}}\text.
\end{equation*}
Assuming that $(\ket{u})_{u \in [\chi]}$ are eigenstates of~$\hat{Q}^{<i}$, then the states~$(\ket{v})_{v \in [\chi]}$ given by
\begin{equation*}
    \ket{v} \coloneq \sum_{u, d_i}\ket{u} (M_i^{d_i})_{u, v}\ket{d_i}
\end{equation*}
are eigenstates of $\hat{Q}^{<i+1}$ if and only if $(M_i^{d_i})_{u, v}$ vanishes for
\begin{equation}
\label{eq:charge_cons}
    q(u)+ q(d_i) \neq q(v)\text.
\end{equation}
Similarly, for $i=1$, the states $\ket{v} = \sum_{d_1}(M_1^{d_1})_v \ket{d_1}$ are eigenstates of $\hat{Q}_1$ iff $(M_1^{d_1})_v \neq 0$ only for $d_1, v$ with $q(d_1) = q(v)$.
By constraining the tensors accordingly, one can thus iteratively create MPS that are eigenstates of $\hat{Q}$.
Thinking about the indices of $M_i$ as carrying a charge~$q$, where $d_i$ and $u$ have incoming and $v$ outgoing directions, the constraints can be interpreted as a charge conservation rule, with $q(u)$ indicating how much charge the state~$\ket{u}$ has acquired in the first $(i-1)$~sites.
At the end, we obtain an MPS~$\ket{\Phi}$ with quantum numbers~$q(\Phi) = Q$ that can be selected by discarding all sequences of quantum numbers that do not ultimately lead to $Q$.
The restriction~\eqref{eq:charge_cons}, together with the constraint set by $Q$, strongly reduces the memory footprint of the MPS~$\ket{\Phi}$.
These resource savings can be further increased by other global abelian symmetries, such as an abelian point group symmetry, as long as all symmetries are compatible.

In practical implementations, all non\hyp zero elements~$(M_i^{d_i})_{v, u}$ whose indices carry the same charges~$q(v)$, $q(d_i)$, $q(u)$ and fulfill $q(u) + q(d_i) = q(v)$ are combined in a block that carries the block index~$(q(v), q(d_i), q(u))$, giving the tensor~$M_i$ a block\hyp sparse structure.
Looking at the corresponding matrices~$M_i^{d_i}$, the charge conservation implies that there is only a single non\hyp zero block in each block row and block column.
We note that it was shown in~\cite{Bachmayr2022} that, by exploiting the gauge degree of freedom, any MPS that is an eigenstate of an operator like $Q$ may be written in terms of block\hyp sparse tensors~$M_i$ as above.

We now want to prepare an MPS~$\ket{\Phi}$ whose tensors~$M_i$ are block\hyp sparse.
As explained in \cref{sec:introduction}, this task reduces to the implementation of the unitaries~${U_i =\begin{bmatrix} U'_i & * \end{bmatrix}}$, where $U_i' \coloneq \begin{bmatrix} M_i^{1} & \dotsm & M_i^{4} \end{bmatrix}^{\!\top}$ is a $(4\chi)\times\chi$~matrix consisting of the stacked matrices~$(M_i^{1})^T, \dotsc, (M_i^{4})^T$, and $*$ is an unspecified block matrix with shape $(4\chi)\times(3\chi)$ that may be chosen freely as long as the overall matrix~$U_i$ is unitary.
The columns of $U_i'$ are orthonormal due to the right\hyp canonical form of $\ket{\Phi}$ and inherit the block\hyp sparse structure of $M_i$, with each block row containing just a single non\hyp zero block.
As shown in the first two boxes of \cref{fig:transform_to_block_diagonal}, we may therefore find a row permutation matrix~$P_r$ that merges the blocks of a block column, transforming the blue non-zero blocks of $U'\coloneq U_i'$ into a chain of rectangles.
The rectangles are at least as high as they are wide and can be expanded into unitary blocks by complementary rectangles~$R$.
In order to obtain a block\hyp diagonal matrix consisting of these unitary blocks, we write the complementary rectangles~$R$ into the unspecified part~$*$ of the matrix, as shown in red in the second box of \cref{fig:transform_to_block_diagonal}, and apply a column permutation matrix~$P_c$.
With a final permutation~$P_o$ of both the rows and columns, the blocks can be ordered according to their sizes, resulting in a block\hyp diagonal matrix~$V$ and a decomposition
\begin{equation*}
    \begin{bmatrix} U' & * \end{bmatrix} = W V Q^{\dagger}
\end{equation*}
for permutation matrices~$W\coloneq P_rP_o$ and $Q\coloneq P_c P_o$.

We implement the block\hyp diagonal matrix~$V$ of dimension $l\coloneqq 4\chi$ by using the techniques from \cref{sec:unitary_synthesis}.
Denoting the unitary blocks of $V$ by $V^1, \dots, V^k$, each block~$V^j = D^j L^j_{l_j}\cdots L^j_{1}$ may be decomposed according to equation~\eqref{eq:lay} into $R_z$\hyp $R_y$ layers~$L_r^j$ with a final diagonal phase matrix~$D^j$.
The number of layers~$l_j$ is given by the size of the unitary block~$V^j$, with the layer structures alternating between the shifted and unshifted versions of~\eqref{eq:form}.
We have $l=\sum_{j} l_j$.
The goal is to combine all $r$-th layers $(L^j_r)_{j\in [k]}$ of the different blocks into a single layer~$L_r$, giving a decomposition
\begin{equation*}
    V=D L_{s}\dotsm L_1\text.
\end{equation*}
If a block $V_j$ does not have an $r$-th layer, we insert the identity, making sure that $L_r$ has dimension $l$.
In order to use the circuits from~\cref{sec:unitary_synthesis}, each layer $L_r$ needs to be either shifted or unshifted.
To make sure that this is the case, we add an identity layer before $L_1^j$ and subsequently shift the indices in the decomposition of $V_j$ for all $j\in [k]$ for which the initial layer $L_1^j$ does not fit the format of the full layer $L_1$ determined by $L_1^1$.
The other layers are then aligned automatically.
Due to this alignment process, the total number~$s$ of layers in the final decomposition can increase from $l_1$, i.e., the block size of the largest block $V_1$, to $l_1+1$ in case there is more than one block of size~$l_1$ and the first layer of one of these blocks does not fit the alignment structure of $L_1$.

As the blocks are ordered according to their sizes, the number of angles to be loaded for the $r$-th layer is $a_r/2$ with $a_r\coloneqq \sum_{\tilde{l}_j> r}l_j$, where we use $\tilde{l}_j$ to denote the length of the decomposition of $V^j$ with a possible alignment filler layer taken into account.
The bitsize of the loaded angles is $2b$ in the general case and $b$ in the real case.
The total Toffoli count of implementing a complex $V$, without the sign fixes and using individual Toffoli/ancilla trade-off parameters~$\Lambda_r$ for each layer, is then upper bounded by
\begin{multline*}
    C_V \coloneq \sum_{r=1}^{s} \left(\biggl \lceil \frac{a_r}{2\Lambda_r} \biggl \rceil + \biggl\lceil \log\left(\frac{l}{a_r}\right)\biggl \rceil + 2b\Lambda_r \right)+ \biggl \lceil \frac{l}{\Lambda} \biggr \rceil \\ + b\Lambda + s\lceil \log(l)\rceil
\end{multline*}
with the cost mostly determined by the size of the largest block.
The terms $\lceil \log(l/a_r) \rceil$ are needed since, even though only $a_r/2$ angles are loaded via a QROAM call with $\lceil \log(a_r/2) \rceil$ address qubits, we need to ensure that no unwanted values are written into the angle registers by controlling the QROAM on the remaining $\lceil \log(l/2) \rceil - \lceil \log(a_r/2) \rceil  \leq \lceil \log(l/a_r) \rceil$ qubits.
Each control qubit requires one Toffoli.
In the real case, the terms with $\Lambda$, for the final phase-rotation, are dropped, and $2b$ is replaced with $b$ in the big parenthesis.

The next step is to find circuits for the permutation matrices $W$ and $Q$.
Instead of implementing a permutation $\ket{i} \mapsto \ket{P(i)}$ exactly, we give ourselves some freedom by only demanding the mapping to be correct up to known sign flips, i.e., we implement $\ket{i} \mapsto \pm \ket{P(i)}$, for which the sign flips can be inferred during runtime.
A circuit for such a permutation on a $v$\hyp qubit register is shown in \cref{fig:permutation}.
Given a basis state~$\ket{i}$, it proceeds as
\begin{equation*}
    \ket{i} \xmapsto{\text{QROAM}} \ket{i}\ket{P(i)} \xmapsto{\text{SWAP}} \ket{P(i)}\ket{i} \xmapsto{\text{M}_X} \pm \ket{P(i)}
\end{equation*}
by first using QROAM to write the target basis state~$\ket{P(i)}$ of $\ket{i}$ in an ancilla register, and subsequently applying a SWAP between the ancilla and working registers.
Measuring the ancilla register, and the QROAM junk registers, in the $X$~basis then leads to the working register being in the state~$\pm \ket{P(i)}$ as desired, with the sign determined by the measurement results, as explained in Appendix~\labelcref{sec:measurement_based_uncomputation}.
The Toffoli cost of the construction is the cost of the QROAM call, which is $2^v/\Lambda'' + v(\Lambda''-1)$, with $\Lambda''$ being a power of two that controls the Toffoli/ancilla trade-off.
In our case, $v=\lceil \log(l)\rceil = \lceil \log(\chi)\rceil + 2$, resulting in a Toffoli count of
\begin{equation*}
    C_P \coloneq \frac{2^{\lceil \log(l)\rceil}}{\Lambda''} + \lceil \log(l)\rceil(\Lambda'' -1)\text.
\end{equation*}
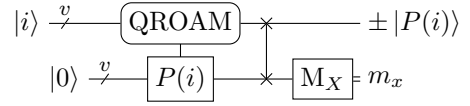
\begin{figure}[t]
    \centering
    \begin{tikzpicture}
        \begin{yquant}[register/separation=1.5mm]
            qubit {$\ket{i}$} states;
            ["north:$v$" {font=\protect\footnotesize, inner sep=1pt}]
            slash states;
            hspace {0.2cm} -;
            [after=states] qubit {$\ket{0}$} ancillas;
            ["north:$v$" {font=\protect\footnotesize, inner sep=1pt}]
            slash ancillas;
            [name = i, operator style={only at={0}{rounded corners}}] box {\Ifnum\idx<1{QROAM}\Else$P(i)$\Fi} states, ancillas;
            swap (states, ancillas);
            box {M$_X$} ancillas;
            settype {cbit} ancillas;
            align -;
            text {$m_x$} ancillas;
            discard ancillas;
            text {$\pm \ket{P(i)}$} states;
            discard states;
        \end{yquant}
        \draw (i-1) -- (i-0);
    \end{tikzpicture}
    \caption{
    Circuit implementing a permutation~$P$ on the basis states of a $v$\hyp qubit quantum register up to sign flips.
    The sign flips can be inferred from the $X$\hyp measurement results~$m_x$ of the ancilla register and the QROAM junk registers not shown.
    }
    \label{fig:permutation}
\end{figure}

\makeatletter
\let\oldarstrut=\@arstrut
\begin{table*}[t]
    \centering
    \begin{adjustbox}{width=\textwidth}
        \begin{tabular}{|cc||c|c|c|c|c|c|c|c|c|c|}
            \hline
            \multicolumn{2}{|c||}{\multirow{2}{*}{System}} & \multicolumn{2}{c|}{P450} & \multicolumn{2}{c|}{P450} & \multicolumn{2}{c|}{[Fe$_2$S$_2$]$^{-2}$} & [Fe$_2$S$_2$]$^{-3}$ & \multicolumn{2}{c|}{[Fe$_4$S$_4$]$^{-2}$} & [Fe$_4$S$_4$]$^{-4}$\\
            & & \multicolumn{2}{c|}{($47$e, $43$o)} & \multicolumn{2}{c|}{($63$e, $58$o)} &\multicolumn{2}{c|}{($30$e, $20$o)} &($31$e, $20$o) &\multicolumn{2}{c|}{($54$e, $36$o)} &($52$e, $36$o)\\
            \hline
            \multicolumn{2}{|c||}{Spin} & $5/2$ & $1/2$ & $5/2$ & $1/2$ & \multicolumn{2}{c|}{$0$} & $1/2$ & \multicolumn{2}{c|}{$0$} & $0$\\
            \hline
            \multirow{2}{30pt}{\;Bond\\ \;\;dim.}& CSF & $1500$ & $1500$ & $1500$ & $1500$ & $8000$ & $1000$ & $8000$ & $8000$ & $1000$ & $8000$\\
            \cline{2-2}
            & DET & $8270$ & $4565$ & $8278$ & $4751$ & $21388$ & $3302$ & $18541$ & $51827$ & $6832$ & $59408$\\
            \hline
            \multicolumn{2}{|c||}{Density$^{-1}$} & $34.5$ & $22.6$ & $32.8$ & $22.5$ & $23.5$ & $23.3$ & $22.7$ & $44.6$ & $42.5$ & $50.2$\\
            \hline
            \multirow{4}{*}{Qubits}& dense & $1062$ & $1026$ & $1074$ & $1034$ & $1985$ & $545$ & $1981$ & $2017$ & $1052$ & $2017$ \\
            \cline{2-2}
            & dense real & $1011$ & $590$ & $1013$ & $620$ & $1025$ & $544$ & $1023$ & $2016$ & $576$ & $2016$ \\
            \cline{2-2}
            & sparse & $590$ & $590$ & $620$ & $620$ & $1025$ & $544$ & $1024$ & $2016$ & $576$ & $2016$ \\
            \cline{2-10}
            \hline
            \multirow{4}{*}{Toffolis}& dense & $8.2\times10^8$ & $4.3\times10^8$ & $1.1\times10^9$ & $6.2\times10^8$ & $1.2\times10^9$ & $8.3\times 10^7$ & $9.0\times 10^8$ & $1.1\times 10^{10}$ & $5.3\times 10^8$ & $1.1\times 10^{10}$\\
            \cline{2-2}
            & dense real & $5.7\times10^8$ & $3.0\times10^8$ & $7.5\times10^8$ & $4.4\times10^8$ & $8.1\times10^8$ & $6.2\times10^7$ & $6.3\times10^8$ &$8.1\times10^9$ & $3.6\times 10^8$ & $8.1\times 10^9$\\
            \cline{2-2}
            & sparse & $3.3\times10^7$ & $2.5\times10^7$ & $4.7\times10^7$ & $3.8\times10^7$ & $6.5\times10^7$ & $5.4\times10^6$ & $5.5\times10^7$ & $3.7\times10^8$ & $2.0\times10^7$ & $3.8\times 10^8$\\
            \cline{2-12}
            & improv. & ${}\times 25$ & ${}\times 17$ & ${}\times 23$ & ${}\times 16$ & ${}\times 18$ & ${}\times 15$ & ${}\times 17$ & ${}\times 30$ & ${}\times 27$ & ${}\times 29$\\
            \hline
            \hline
            \multirow{2}{35pt}{\;QPE-\\ cost~\cite{Low2025}}& Qubits & - & - & $1150$ & $1150$ & \multicolumn{2}{c|}{$463$} & - & \multicolumn{2}{c|}{$868$} & -\\
            \cline{2-2}
            & Toffolis & - & - & $4.9\times 10^8$ & $4.9\times 10^8$ & \multicolumn{2}{c|}{$4.0\times 10^7$} & - & \multicolumn{2}{c|}{$1.7\times 10^7$} & -\\
            \cline{2-2}
            \hline
        \end{tabular}
    \end{adjustbox}
    \caption{
    Resource estimates for the preparation of MPS from DMRG calculations of some strongly correlated molecular systems.
    The scripts and other assets for reproducing the results are available on Zenodo~\cite{Rupprecht2026Zenodo}.
    The molecule data of the iron\hyp sulfur clusters are taken from \cite{ollitrault2024data} and are based on the works~\cite{Lee2023, Ollitrault2024, Sharma2014}.
    The molecule data for the active spaces of the P450\hyp enzymes were taken from~\cite{goings2022data} based on the paper~\cite{Goings2022}.
    The bond dimensions of both the spin\hyp adapted MPS version (\emph{CSF}) used for the DMRG calculations and the non--spin\hyp adapted version (\emph{SOS}), that is used for the preparation, are provided.
    For context we print the inverse of the density of the tensors, i.e., the quotient of the dense tensor sizes and the number of non-zero elements in the tensors.
    We show the Toffoli and qubit estimates as computed by \texttt{Qualtran} for the baseline scheme of~\cite{berry2024} without (\emph{dense}) and with (\emph{dense real}) the unitary synthesis improvement for real unitaries of \cref{sec:unitary_synthesis}, and for our sparse method (\emph{sparse}).
    We compute the \emph{improvement factors} from the dense to the sparse method.
    The angle bitsize~$b$ is set to~$15$ in all simulations.
    For context, we print resource estimates for running QPE\hyp based state\hyp of\hyp the\hyp art ground\hyp state energy estimation taken from Table~V in~\cite{Low2025}.
    }
    \label{tab:ressource_counts}
\end{table*}
\let\@arstrut=\oldarstrut

For realistic bond dimensions, e.g., $\chi = 1000$, this cost is of a similar magnitude as the cost of a single layer of $V$.
Indeed, in the case of a real unitary, the QROAM of such a layer loads up to $l/2 = 2\chi $ elements of bitsize $b \approx 15$, while the permutation loads $2^{\lceil \log(l)\rceil} \approx 4\chi$ elements with bitsizes given by $\lceil \log(l)\rceil = \lceil \log(\chi)\rceil + 2 = 12$.
Generally, whenever there are multiple gadgets in a quantum circuit that independently allow for Toffoli/qubit trade-offs, one should check whether there are individual gadgets that require a disproportionate number of ancilla qubits.
If those gadgets only contribute mildly to the total Toffoli count, it is sensible to adapt the corresponding trade-off parameter, leading to more favorable qubit requirements.

What remains is to look at the sign flips.
As discussed in~\cref{sec:unitary_synthesis}, the sign flips of the initial permutation~$W$ can be commuted through $V$ by dynamically adapting the rotation angles of each layer $L_r$.
The combined signs from both $W$ and $V$ can then be pushed through $Q^\dagger$ by noting that $s_j\ket{j}$ gets mapped to $s_j\ket{Q^\dagger(j)}$.
Including the sign flips from $Q^{\dagger}$ itself, we can apply a standard sign\hyp fixup routine~\cite{Berry2019} after $Q^\dagger$ with Toffoli cost
\begin{equation*}
    C_F\coloneq \frac{2^{\lceil \log(l)\rceil}}{\Lambda'} + \Lambda'
\end{equation*}
where $\Lambda'$ is, as usual, a power of two controlling the Toffoli/qubit trade-off.
The overall cost for the implementation of $\begin{bmatrix} U' & *\end{bmatrix}$ is then $C_V + 2C_P + C_F$.

To benchmark our scheme, we choose some strongly correlated molecular systems used in various papers about ground\hyp state energy estimation~\cite{Low2025, Goings2022} and initial state preparation~\cite{berry2024, Lee2023, Ollitrault2024}.
The results are summarized in \cref{tab:ressource_counts} with the scripts and other assets for reproducing them available on Zenodo~\cite{Rupprecht2026Zenodo}.
We run spin\hyp adapted DMRG calculations with the \texttt{block2} library~\cite{zhai2023} and transform the resulting MPS from spin\hyp adapted form into non\hyp spin\hyp adapted form, with the latter generally having a larger maximal bond dimension.
The maximal bond dimensions used for the calculations are chosen as in~\cite{Ollitrault2024}, with additional runs for some of the Iron-Sulfur clusters with the bond dimension chosen as in~\cite{berry2024}.
For comparison, in Appendix~\labelcref{sec:additional_ressource_estimates} we give resource estimates for MPS preparation where the bond dimension of the non\hyp spin\hyp adapted version is truncated to the spin\hyp adapted value.
For systems with non\hyp zero spin, the DMRG calculation is done within a singlet embedding~\cite{sharma2012}, and the non\hyp spin\hyp adapted form is a combination of different spin projection states, leading to $B_0$, as defined in~\cref{sec:introduction}, being greater than one.
This can easily be taken into account by adapting the state preparation~$U_1$, for which we use the Grover–-Rudolph variant described in~\cite{Rupprecht2026}.
We extract the tensors with \texttt{pyblock3}~\cite{zhai} and compute the resource estimates with \texttt{Qualtran}.
After visual inspection we set the Toffoli/ancilla parameters of the QROAMs such that outliers with disproportional high ancilla-qubit numbers are removed.

As a baseline, we use the MPS preparation scheme of Berry et al.~\cite{berry2024}, in which the matrices~$U_i$ are decomposed into blocks of unitaries of size~$\chi$, as recalled in Appendix~\labelcref{sec:details_on_berry_paper}.
We run the Berry et al. scheme with and without the unitary synthesis improvements from \cref{sec:unitary_synthesis}, validating the $\sqrt{2}$ speedup.
Exploiting the block\hyp sparsity with our preparation scheme, we observe Toffoli count improvement factors of around $10 - 30$ (cf. also Appendix~\labelcref{sec:additional_ressource_estimates}) compared to the baseline.
We expect similar improvements for other systems.
The speed-up is greater for MPS of higher bond dimension where the density of non-zero entries in the tensors is lower.
If we were to perfectly exploit the sparsity of the MPS tensors, taking the $\sqrt{2}$-factor from the improved unitary synthesis into account, we would expect a speedup of $\sqrt{2}$ times the inverse of this density.
In the examples we generally achieve roughly half of this value.

For some of the molecules, resource estimates for running state-of-the-art QPE-based ground-state energy estimation were given in~\cite{Low2025}.
We print these values for context, noting that the cost of the MPS preparation can be further adapted by compressing the MPS via reduction of the bond dimension, often with only little difference in the quality of the state as seen in Appendix~\labelcref{sec:additional_ressource_estimates}.
In order to get an estimate on the general quality of an MPS as initial state for a QPE calculation, i.e., high overlap with the desired ground-state, the empirical method proposed in~\cite{berry2024} may be used.
As we do not discuss these topics in more detail, the estimates in \cref{tab:ressource_counts} should only be taken as a benchmark of the effectiveness of our preparation approach on a range of MPS with different bond dimensions and not as a table of initial state preparation costs for the respective molecules.

\section{Conclusion and Outlook}
\label{sec:conclusion}

In this work, we have developed improved methods for preparing matrix product states with symmetry\hyp induced block\hyp sparse structure commonly encountered in quantum chemistry and condensed\hyp matter systems.
In those systems, $U(1)$\hyp symmetries from particle number and spin projection conservation result in the MPS tensors being block\hyp sparse, which is widely used in classical algorithms such as DMRG.
We exploit the block\hyp sparsity within the standard ancilla\hyp assisted linear\hyp depth preparation method of~\cite{Schoen2005} by implementing row and column permutations that transform the block\hyp sparse matrices into block\hyp diagonal form.
These block\hyp diagonal matrices can then be synthesized, with the cost being determined by the size of the largest block.
To this end, we also improve the Toffoli count of unitary synthesis for real unitaries within the approach of~\cite{berry2024} by a factor of $\sqrt{2}$.

When transforming the block\hyp sparse matrix into block\hyp diagonal form, we need to insert non\hyp trivial values in the unspecified part of the matrix in order to expand the rectangles into unitary blocks.
From a parameter counting perspective, these values are not needed and cause costs when loading the respective angles.
It is therefore an interesting question whether more efficient schemes for the blocking step, together with the subsequent unitary synthesis, are possible, e.g., via ideas from~\cite{kottmann2026}.
Another direction to pursue is the adaptation of our scheme for MPS in spin-adapted $\mathit{SU}(2)$ form~\cite{SinghSU22012}, where the computational basis states correspond to configuration state functions instead of Slater determinants. 

Applying the new MPS preparation techniques to various systems from quantum chemistry, we see Toffoli cost improvement factors of around $10 - 30$ compared to the state of the art~\cite{berry2024}.
As due to recent algorithmic improvements for ground\hyp state energy estimation~\cite{Low2025}, the cost of preparing an initial state can be of similar magnitude as the phase estimation task itself, this is relevant for allowing potential QPE improvements to bring down the end\hyp to\hyp end cost of ground\hyp state energy estimation further.
It moreover allows the preparation of higher-quality initial states.

In this context, it is an interesting question how the direct preparation of MPS compares to computing a sparse sum of Slater determinants (SOS) representation from the MPS and preparing this SOS state directly~\cite{Fomichev2024, Rupprecht2026}.
Similarly, comparisons to techniques of adiabatic~\cite{Lee2023, han2026} or dissipative~\cite{motlagh2024, Li2025, Ding2024} nature would be desirable.
For the Fermi--Hubbard model, which is often seen as particularly suitable for quantum simulation due to efficient methods for time evolution~\cite{Kan2025, huggins2025}, the cost of preparing suitable initial states has not yet received much attention.
To this end, it would be interesting to see whether our approach allows for preparing MPS of sufficient quality~\cite{liu2025} without changing the magnitudes of the resource requirements for the end\hyp to\hyp end algorithm.

The algorithms in this paper are implemented as \texttt{Qualtran} bloqs, with the more resource-intensive task of finding the Givens decompositions of the unitary matrices outsourced into Rust code using the \texttt{faer} linear algebra library~\cite{Kazdadi2026}.
The runtime of this task may be further improved by a GPU implementation.
Similarly, the implementation of the block-diagonalization procedure explained in~\cref{sec:mps_preparation}, using general sparse arrays at the moment, may profit from directly working on the data structures from the DMRG-libraries such as \texttt{pyblock3}.

In conclusion, we expect our MPS preparation approach to be an integral part in connecting classical and quantum simulation of systems from chemistry and condensed\hyp matter physics.

\section*{Data availability statement}
All code and other assets prepared for the paper are available on Zenodo~\cite{Rupprecht2026Zenodo}.

\section*{Acknowledgements}

We thank Benjamin Desef for valuable discussions.
This work has been funded by the Ministry of Economic Affairs, Labour and Tourism Baden\hyp Württemberg through the Competence Center Quantum Computing Baden\hyp Württemberg (KQCBW).
The circuit diagrams in the paper were created with the \texttt{yquant} package~\cite{desef2021}.

Open-weight large language models (LLMs) were used for language checking and for minor assistance with coding.
Moreover, a part of the proof of the Lemma in Appendix~\labelcref{sec:details_on_berry_paper} was aided by a hint from an LLM.

\bibliography{main.bbl}
\onecolumngrid
\appendix
\newpage

\section{Measurement-based uncomputation}
\label{sec:measurement_based_uncomputation}
Throughout the paper, we often use measurement\hyp based uncomputation for resetting ancilla qubits, allowing us to avoid physical uncomputation of those ancillas~\cite{berry2024, Rupprecht2026, gidney2025, babbush2026}.
In this section, we briefly recall this technique.

Given a state~$\sum_k a_k \ket{k}\ket{f(k)}^a$, where $\ket{f(k)}^a$ is a computational basis state in an ancilla register belonging to the basis state~$\ket{k}$, the qubits of the ancilla register are measured in the $X$~basis and discarded afterwards.
Recall that $\ket{0} = (\ket{+}+\ket{-})/\sqrt{2}$ and $\ket{1} = (\ket{+}-\ket{-})/\sqrt{2}$.
Assuming for now that the ancilla register consists of a single qubit, the state after the measurement process is $\sum_k a_k \ket{k}$ if $\ket{+}^a$ was measured, and $\sum_k (-1)^{f(k)} a_k \ket{k}$ if the measurement yields $\ket{-}^a$.
Generalizing this to multiple ancilla qubits results in the state~$\sum_k (-1)^{F(k)} a_k \ket{k}$, where $F(k) \coloneq \sum_{i\in I} f_i(k)$ sums over the values~$f_i(k)$ of $f(k)$ at the bits~$i \in I$ for which the measurement result of the corresponding ancilla qubit was $\ket{-}$.

The sign flips inferred from the measurement results can then be taken into account in subsequent computations and, if required, be corrected using standard methods~\cite{kottmann2026, Berry2019}.
When multiple registers need to be uncomputed at different times, it is sensible to try and commute the phase flips to a point in the circuit where the joint sign corrections can be applied at once, instead of fixing them directly after each measurement.
When doing so, the parts of the circuit that come after a measurement might need to be adapted in order to commute the sign flips through.
Since the measurement results are only available during runtime, one needs to be able to make the necessary adaptations on the fly, which means that the logic for finding those adaptations should not require large amounts of classical computation.

\section{Details on the Berry et al. MPS preparation method}
\label{sec:details_on_berry_paper}

In this section, we give some details on our implementation of the MPS preparation scheme from the paper~\cite{berry2024}.
In the main part, this is referred to as the dense preparation method.
We repeatedly make use of the following lemma.
\begin{lemma}
    Let $A, B \in \C^{m \times k}$ be complex $m\times k$~matrices such that the columns in the block matrix~$
    \begin{bsmallmatrix}
        A\\
        B
    \end{bsmallmatrix}
    $
    are orthonormal, then there exist unitary matrices~$U_1, U_2 \in \C^{m \times m}$, a unitary~$V\in \C^{k \times k}$, and real diagonal matrices~$D_1, D_2 \in \R^{m \times k}$ such that
    \begin{equation}
    \label{eq:dec_berry}
        \begin{bmatrix}
            A\\
            B
        \end{bmatrix}
        =
        \begin{bmatrix}
            U_1 & 0\\
            0 & U_2
        \end{bmatrix}
        \begin{bmatrix}
            D_1 & -D_2\\
            D_2 & D_1
        \end{bmatrix}
        \begin{bmatrix}
            V\\
            0
        \end{bmatrix}
    \end{equation}
    and the block matrix~$
    D \coloneq
    \begin{bsmallmatrix}
        D_1 & -D_2\\
        D_2 & D_1
    \end{bsmallmatrix}
    $ in the middle has orthonormal columns.
    If $A$ and $B$ are real, all matrices in the decomposition can be made real.
\end{lemma}
\begin{proof}
    From a singular value decomposition of $A$, we get unitary matrices~$U_1 \in \C^{m \times m}$ and $V \in \C^{k \times k}$ together with a real diagonal matrix~$D_1 \in \R^{m \times k}$ such that $A=U_1 D_1 V$.
    Moreover, a polar decomposition of the matrix~$BV^{\dagger} \in \C^{m \times k}$ yields $\tilde{U}_2\in \C^{m \times k}$ with orthonormal columns and a positive semi\hyp definite hermitian matrix~$\tilde{D}_2 \in \C^{k \times k}$ such that $B=\tilde{U}_2\tilde{D}_2V$.
    Expanding $\tilde{U}_2$ to a unitary~$U_2\in \C^{m \times m}$ and padding $\tilde{D}_2$ with zeroes to a matrix~$D_2 \in \C^{m \times k}$, we get $B=U_2 D_2 V$.
    From the orthonormality of the columns of $
    \begin{bsmallmatrix}
        A\\
        B
    \end{bsmallmatrix}$,
    one has
    \begin{equation*}
        I = A^{\dagger}A + B^{\dagger}B =V^{\dagger} D_1^{\dagger} D_1 V + V^{\dagger} D_2^{\dagger} D_2 V
    \end{equation*}
    and hence $\tilde{D}_2^2 = \tilde{D}_2^{\dagger}\tilde{D}_2 = D_2^{\dagger}D_2 =  I-D_1^{\dagger}D_1$, since $\tilde{D}_2$ is hermitian.
    As $D_1$ is diagonal with real entries, so is $\tilde{D}_2$ and consequently $D_2$, and it only remains to show that the middle matrix has orthonormal columns.
    This can be verified by direct computation of $D^{\dagger}D$, making use of the fact that $D_1$ and $D_2$ are diagonal and $D_2^{\dagger}D_2 =  I-D_1^{\dagger}D_1$.
\end{proof}

The lemma guarantees the existence of a decomposition as in equation (30) of reference~\cite{berry2024}.
In the paper, it is claimed that (for $m=k$) a QR decomposition of $BV^{\dagger}$ yields the desired $D_2$, however, if $B$ does not have full rank, the $D_2$ from the QR decomposition in general is not diagonal.
We fix this by replacing the QR decomposition by a polar decomposition.

Given an MPS $\ket{\Phi}$, we consider a site $i$ of local dimension $4$ and want to find a sequence of unitaries that implements~$U_i = \begin{bmatrix} U_i' & *\end{bmatrix}$ as defined in \cref{sec:introduction}, where \smash{$U_i' = \begin{bmatrix} M_i^{1} & \cdots & M_i^{4} \end{bmatrix}^{\!\raisebox{-1pt}{\ensuremath{\scriptstyle\top}}}$} and $*$ is a block matrix that can be freely chosen as long as $U_i$ is unitary.
For ease of notation, we define $A_{j}\coloneq (M_i^{j})^\top$.
Diverting from the main text and~\cite{berry2024}, we remove the assumption that the incoming and outgoing bond dimensions of $M_i$ are both equal to $\chi$ and instead consider the $A_j$ as matrices in $\C^{m \times k}$.
Applying suitable padding if necessary, we assume $m\geq k$.
The scheme in \cref{sec:mps_preparation} works without change in this more general situation.
The idea of the paper~\cite{berry2024} is to expand the matrix $U'\coloneq U_i'$, consisting of orthonormal columns, into a unitary by writing it as a series of unitaries on smaller subspaces.
We start by decomposing the stacked block matrices
\begin{align*}
    U'=
    \begin{bmatrix}
        A_1\\
        \red{A_2}\\
        \red{A_3}\\
        \red{A_4}
    \end{bmatrix}
    =
    \begin{bmatrix}
        I & 0\\
        0 & B_2\\
        0 & \blue{B_3}\\
        0 & \blue{B_4}
    \end{bmatrix}
    \begin{bmatrix}
        A_1\\
        R_2
    \end{bmatrix}
    =
    \begin{bmatrix}
        I & 0 & 0\\
        0 & I & 0\\
        0 & 0 & C_3\\
        0 & 0 & C_4
    \end{bmatrix}
    \begin{bmatrix}
        I & 0\\
        0 & B_2\\
        0 & S_3\\
    \end{bmatrix}
    \begin{bmatrix}
        A_1\\
        R_2
    \end{bmatrix}
\end{align*}
by first applying a QR decomposition to the lower three blocks (red) of $U'$, resulting in blocks $B_{i} \in \C^{m \times m}$ that, stacked on top of each other, have orthonormal columns, and $R_2 \in \C^{m \times k}$, which is an upper triangular matrix.
Analogously, another QR decomposition applied to the lower two blocks (blue) of the $B_i$ block-matrix yields $C_3, C_4 \in \C^{m \times m}$ that, stacked, also have orthonormal columns, and an upper-triangular matrix~$S_{3} \in \C^{m \times m}$.
The columns of the matrix~$
\begin{bsmallmatrix}
    B_2\\
    S_3
\end{bsmallmatrix}
$ are orthonormal, as can be seen by computing
\begin{equation*}
    B_2^{\dagger}B_2 + S_3^{\dagger}S_3 = B_2^{\dagger}B_2 + S_3^{\dagger}(C_3^{\dagger}C_3 + C_4^{\dagger}C_4)S_3 = B_2^{\dagger}B_2 + B_3^{\dagger}B_3 + B_4^{\dagger}B_4 = I
\end{equation*}
and the same holds true for $
\begin{bsmallmatrix}
    A_1\\
    R_2
\end{bsmallmatrix}
$.
We may therefore decompose each of the matrices further by applying the lemma, leading to
\begin{multline*}
    U' =
    \begin{bmatrix}
        A_1\\
        A_2\\
        A_3\\
        A_4\\
    \end{bmatrix}
    =
    \begin{bmatrix}
        I & 0 & 0 & 0\\
        0 & I & 0 & 0\\
        0 & 0 & U_3 & 0\\
        0 & 0 & 0 & U_4
    \end{bmatrix}
    \begin{bmatrix}
        I & 0 & 0 & 0\\
        0 & I & 0 & 0\\
        0 & 0 & D_3 & -D_3'\\
        0 & 0 & D_3' & D_3
    \end{bmatrix}
    \begin{bmatrix}
        I & 0 & 0\\
        0 & I & 0\\
        0 & 0 & V''\\
        0 & 0 & 0
    \end{bmatrix}
    \begin{bmatrix}
        I & 0 & 0\\
        0 & U_2 & 0\\
        0 & 0 & U
    \end{bmatrix}
    \begin{bmatrix}
        I & 0 & 0\\
        0 & D_2 & -D_2'\\
        0 & D_2' & D_2
    \end{bmatrix}
    \begin{bmatrix}
        I & 0\\
        0 & V'\\
        0 & 0
    \end{bmatrix}\\
    \begin{bmatrix}
        U_1 & 0\\
        0 & U'
    \end{bmatrix}
    \begin{bmatrix}
        D_1 & -D_1'\\
        D_1' & D_1
    \end{bmatrix}
    \begin{bmatrix}
        V\\
        0
    \end{bmatrix}.
\end{multline*}
where all matrix blocks are $m \times m$ except $D_1, D_1' \in \R^{m \times k}$ and $V \in \C^{k \times k}$.
Setting $W_1 \coloneq V'U'$, $W_2 \coloneq V''U$, the decomposition
\begin{multline*}
    \begin{bmatrix}
        A_1 & * & * & *\\
        A_2 & * & * & *\\
        A_3 & * & * & *\\
        A_4 & * & * & *\\
    \end{bmatrix}
    =
    \begin{bmatrix}
        U_1 & 0 & 0 & 0\\
        0 & U_2 & 0 & 0\\
        0 & 0 & U_3 & 0\\
        0 & 0 & 0 & U_4\\
    \end{bmatrix}
    \begin{bmatrix}
        I & 0 & 0 & 0\\
        0 & I & 0 & 0\\
        0 & 0 & D_3 & -D_3'\\
        0 & 0 & D_3' & D_3\\
    \end{bmatrix}
    \begin{bmatrix}
        I & 0 & 0 & \blue{0}\\
        0 & I & 0 & \blue{0}\\
        0 & 0 & W_2 & \blue{0}\\
        \blue{0} & \blue{0} & \blue{0} & \blue{W_3}\\
    \end{bmatrix}
    \begin{bmatrix}
        I & 0 & 0 & \blue{0}\\
        0 & D_2 & -D_2' & \blue{0}\\
        0 & D_2' & D_2 & \blue{0}\\
        \blue{0} & \blue{0}& \blue{0}&\blue{I}
    \end{bmatrix}\\
    \begin{bmatrix}
        I & 0 & \blue{0} & \blue{0}\\
        0 & W_1 & \blue{0} & \blue{0}\\
        0 & 0 & \blue{W_1} & \blue{0}\\
        \blue{0} & \blue{0} & \blue{0} & \blue{0}
    \end{bmatrix}
    \begin{bmatrix}
        D_1 & -D_1' & \blue{0} & \blue{0}\\
        D_1' & D_1  & \blue{0} & \blue{0}\\
        \blue{0} & \blue{0} &\blue{D_1} & \blue{-D_1'}\\
        \blue{0} & \blue{0} & \blue{D_1'} & \blue{D_1}
    \end{bmatrix}
    \begin{bmatrix}
        V& \blue{0} & \blue{0} & \blue{0}\\
        0 & \blue{V} & \blue{0} & \blue{0}\\
        \blue{0} & \blue{0} & \blue{V} & \blue{0}\\
        \blue{0} & \blue{0} & \blue{0} & \blue{V}
    \end{bmatrix}
\end{multline*}
reproduces the block column containing the $A_i$.
Here, we chose the added, free (blue) entries such that the overall matrix is unitary and the cost of implementation is favorable.
The blocks of the matrices correspond to the four subspaces defined by the two site qubits where we assume the registers to be arranged with the site register first and the $\lceil \log(m) \rceil$-qubit ancilla register second.
The unitary~$V$ does not couple different subspaces and only acts on the ancilla register itself.
It hence can be pushed to the unitary of the prior site.
Doing so, and implicitly expanding the block matrix $\begin{bsmallmatrix}
    D_1 & -D_1'\\
    D_1' & D_1\\
\end{bsmallmatrix}$ to consist of $m\times m$~blocks, one arrives at the circuit in \cref{fig:decomposition_berry}.
\begin{figure}[ht]
    \centering
    \begin{tikzpicture}
        \begin{yquant}[every control/.append style={radius=0.8mm}, register/separation=2.5mmm]
            qubit {site[0]} site0;
            qubit {anc} anc;
            qubit {site[1]} site1;
            ["north:$w$" {font=\protect\footnotesize, inner sep=1pt}]
            slash anc;
            box {$
            \begin{bmatrix}
                D_1 & -D_1'\\
                D_1'& D_1
            \end{bmatrix}
            $} (anc, site1);
            cnot site1 | site0;
            box {$W_1$} anc | site1;
            box {$
            \begin{bmatrix}
                D_2 & -D_2'\\
                D_2'& D_2
            \end{bmatrix}
            $} (anc, site0) | site1;
            cnot site1 | site0;
            box {$W_2$} anc | site0;
            box {$
            \begin{bmatrix}
                D_3 & -D_3'\\
                D_3'& D_3
            \end{bmatrix}
            $} (anc, site1) | site0;
            box {$
            \begin{bmatrix}
                U_1 & 0 & 0 & 0\\
                0 & U_2 & 0 & 0\\
                0 & 0 & U_3 & 0\\
                0 & 0 & 0 & U_4
            \end{bmatrix}
            $} (anc, site1, site0);
            box {\rotatebox{90}{SignFixes}} (anc, site0, site1);
        \end{yquant}
    \end{tikzpicture}
    \caption{
    Decomposition of $U_i = \begin{bmatrix} U_i' & * \end{bmatrix}$ into unitaries on smaller subspaces with $V$ moved to the prior site.
    The ancilla register has size $w\coloneq \log(m)$.
    For readability, we move the second qubit of the two-qubit site register below the ancilla register, yielding the order (site[0], anc, site[1]).
    In the matrix representations, it is assumed to be above, i.e., (site, anc).
    The individual unitaries are implemented up to sign flips, which are then fixed at the end.
    }
    \label{fig:decomposition_berry}
\end{figure}
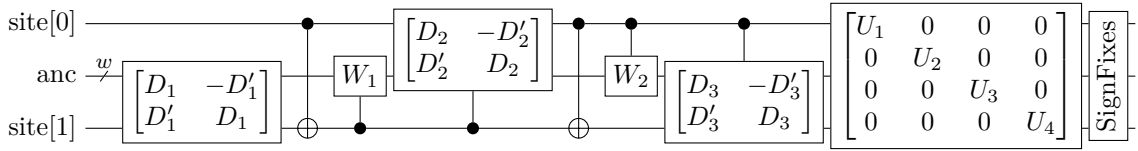

As in the main text, we implement the unitaries only up to known sign flips (cf. Appendix~\labelcref{sec:measurement_based_uncomputation}) and commute those sign flips to the end of the circuit, where they are corrected via a standard construction from~\cite{Berry2019}.
Commuting sign flips through unitaries requires the unitaries to be adapted according to the incoming signs.
In the case of $W_1$ and $W_2$, this is troublesome, as the incoming signs can differ between the two blocks in which each of the unitaries appears.
However, we may use the fact that the lower right of the blocks can actually be freely chosen and can therefore absorb any incoming sign differences.

The overall Toffoli cost for $m =\chi$ is upper bounded by
\begin{equation*}
    3\left( \Bigl \lceil \frac{\chi}{\Lambda} \Bigr \rceil + b\Lambda \right) + 2(\chi +1) \left( \Bigl \lceil \frac{\chi}{2\Lambda'} \Bigr \rceil + 2b\Lambda'\right) + (\chi +1) \left( \Bigl \lceil \frac{2\chi}{\Lambda''} \Bigr \rceil + 2b\Lambda''\right) + 3(\chi\lceil \log(\chi) \rceil + b) + \left( \Bigl \lceil \frac{4\chi}{\Lambda'''} \Bigr \rceil + \Lambda''' \right)
\end{equation*}
where the four big parentheses correspond, in order, to the $D_i$/$D_i'$ rotations, the two $W_i$ unitaries, the $U_i$ unitary, and the sign fix.
The parameters $\Lambda$, $\Lambda'$, $\Lambda''$ and $\Lambda'''$ are powers of two that control the Toffoli/qubit trade-offs.
The Toffoli count for real MPS, using our improved unitary synthesis method introduced in \cref{sec:unitary_synthesis}, can be calculated by replacing all occurrences of $2b$ by $b$ and all $(\chi +1)$ by $\chi$.
Moreover, the $b$ in the small parentheses can be removed.

\section{Resource estimates for truncated MPS}
\label{sec:additional_ressource_estimates}

In this section we complement~\cref{tab:ressource_counts} by providing resource estimates for the case in which the bond dimension of the MPS after the transformation to the non-spin-adapted form is truncated to the bond dimension of the original spin-adapted version.
As one can see, the truncated MPS often differ only slightly from the non-truncated MPS leading to a reduction of the required resources for their preparation.

\makeatletter
\let\oldarstrut=\@arstrut
\begin{table*}[ht]
    \centering
    \begin{adjustbox}{width=\textwidth}
        \begin{tabular}{|cc||c|c|c|c|c|c|c|c|c|c|}
            \hline
            \multicolumn{2}{|c||}{\multirow{2}{*}{System}} & \multicolumn{2}{c|}{P450} & \multicolumn{2}{c|}{P450} & \multicolumn{2}{c|}{[Fe$_2$S$_2$]$^{-2}$} & [Fe$_2$S$_2$]$^{-3}$ & \multicolumn{2}{c|}{[Fe$_4$S$_4$]$^{-2}$} & [Fe$_4$S$_4$]$^{-4}$\\
            & & \multicolumn{2}{c|}{($47$e, $43$o)} & \multicolumn{2}{c|}{($63$e, $58$o)} &\multicolumn{2}{c|}{($30$e, $20$o)} &($31$e, $20$o) &\multicolumn{2}{c|}{($54$e, $36$o)} &($52$e, $36$o)\\
            \hline
            \multicolumn{2}{|c||}{Spin} & $5/2$ & $1/2$ & $5/2$ & $1/2$ & \multicolumn{2}{c|}{$0$} & $1/2$ & \multicolumn{2}{c|}{$0$} & $0$\\
            \hline
            \multirow{3}{30pt}{\;Bond\\ \;\;dim.}& CSF & $1500$ & $1500$ & $1500$ & $1500$ & $8000$ & $1000$ & $8000$ & $8000$ & $1000$ & $8000$\\
            \cline{2-2}
            & DET & $8270$ & $4565$ & $8278$ & $4751$ & $21388$ & $3302$ & $18541$ & $51827$ & $6832$ & $59408$\\
            \cline{2-2}
            & DET trunc. &  $1500$ & $1500$ & $1500$ & $1500$ & $8000$ & $1000$ & $8000$ & $8000$ & $1000$ & $8000$\\
            \hline
            \multicolumn{2}{|c||}{Overlap} & $0.9988$ & $0.9994$ & $0.9970$ & $0.9981$ & $1.0000$ & $0.9996$ & $1.0000$ & $0.9739$ & $0.8229$ & $0.9663$\\
            \hline
            \multicolumn{2}{|c||}{Density$^{-1}$} & $22.6$ & $16.8$ & $21.4$ & $16.8$ & $20.3$ & $19.2$ & $19.8$ & $33.0$ & $28.3$ & $36.2$ \\
            \hline
            \multirow{2}{*}{Qubits}& dense & $589$ & $589$ & $619$ & $619$ & $1024$ & $304$ & $1024$ & $1056$ & $336$ & $1056$\\
            \cline{2-2}
            & sparse & $349$ & $349$ & $379$ & $379$ & $545$ & $303$ & $545$ & $577$ & $335$ & $577$\\
            \cline{2-10}
            \hline
            \multirow{3}{*}{Toffolis}& dense & $1.1\times10^8$ & $1.0\times10^8$ & $1.5\times10^8$ & $1.5\times10^8$ & $2.1\times10^8$ & $1.2\times 10^7$ & $2.3\times10^8$ & $6.0\times10^8$ & $2.8\times 10^7$ & $6.0\times10^8$\\
            \cline{2-2}
            \cline{2-2}
            & sparse & $5.8\times10^6$ & $7.6\times10^6$ & $8.9\times10^6$ & $1.1\times10^7$ & $2.1\times 10^7$ & $1.2\times10^6$ & $2.2\times10^7$ & $3.2\times10^7$ & $1.9\times10^6$ & $2.9\times10^7$\\
            \cline{2-12}
            & improv. & ${}\times 18$ & ${}\times 14$ & ${}\times 17$ & ${}\times 13$ & ${}\times 10$ & ${}\times 10$ & ${}\times 10$ & ${}\times 18$ & ${}\times 14$ & ${}\times 20$\\
            \hline
            \hline
            \multirow{2}{35pt}{\;QPE-\\ cost~\cite{Low2025}}& Qubits & - & - & $1150$ & $1150$ & \multicolumn{2}{c|}{$463$} & - & \multicolumn{2}{c|}{$868$} & -\\
            \cline{2-2}
            & Toffolis & - & - & $4.9\times 10^8$ & $4.9\times 10^8$ & \multicolumn{2}{c|}{$4.0\times 10^7$} & - & \multicolumn{2}{c|}{$1.7\times 10^7$} & -\\
            \cline{2-2}
            \hline
        \end{tabular}
    \end{adjustbox}
    \caption{Additional resource estimates for the preparation of MPS from DMRG calculations of some strongly correlated molecular systems.
    The systems are the same as in~\cref{tab:ressource_counts} with the difference being that the MPS after the transformation from spin\hyp adapted to non\hyp spin\hyp adapted form are truncated to the bond dimension (\emph{DET trunc.}) of the spin\hyp adapted value.
    We list the overlap between the truncated and non-truncated MPS.
    }
    \label{tab:ressource_counts_truncated}
\end{table*}
\let\@arstrut=\oldarstrut
\end{document}